%% file: icassp_one_column.tex
\documentclass[a4paper,11pt]{article}
\usepackage[margin=1in]{geometry}
\usepackage{amsmath}  
\usepackage{xcolor}
\usepackage{bbm}
\usepackage{changepage}
\usepackage{verbatim}
\input{defpack}

\author{Yiyue Chen,~\IEEEmembership{Student Member,~IEEE,} Abolfazl~Hashemi,~\IEEEmembership{Student Member,~IEEE,} and Haris~Vikalo,~\IEEEmembership{Senior Member,~IEEE,} and 
 \thanks{Yiyue Chen, Abolfazl Hashemi, and Haris Vikalo are with the Department of Electrical and Computer Engineering, 
 University of Texas at Austin, Austin, TX 78712 USA. This work was supported in part by NSF grant 1809327.}
 }

\begin{document}
\title{Decentralized Optimization On Time-Varying Directed Graphs Under Communication Constraints}
\date{}
\author{Yiyue Chen, Abolfazl Hashemi, Haris Vikalo\thanks{Yiyue Chen and Haris Vikalo  are with the Department of Electrical and Computer Engineering, University of Texas at Austin, Austin, TX 78712 USA. Abolfazl Hashemi is with the Oden Institute for Computational Engineering and Sciences, 
 University of Texas at Austin, Austin, TX 78712 USA. This work was supported in part by NSF grant 1809327.}}

\maketitle
\begin{abstract}
	We consider the problem of decentralized optimization where a collection of agents, each having access to a local cost 
function, communicate over a time-varying directed network and aim to minimize the sum of those functions. In practice, 
the amount of information that can be exchanged between the agents is limited due to communication constraints. We 
propose a communication-efficient algorithm for decentralized convex optimization that rely on sparsification of local 
updates exchanged between neighboring agents in the network. In directed networks, message sparsification alters 
column-stochasticity -- a property that plays an important role in establishing convergence of decentralized learning tasks. 
We propose a decentralized optimization scheme that relies on local modification of mixing matrices, and show that it 
achieves $\mathcal{O}(\frac{\mathrm{ln}T}{\sqrt{T}})$ convergence rate in the considered settings. Experiments validate theoretical
results and demonstrate efficacy of the proposed algorithm.
	
\end{abstract}

\section{Introduction}\label{sec:intro}
\input{icassp_intro}
\section{Problem Setting}\label{sec:problem}
\input{icassp_problem}
\section{Compressed Time-Varying Decentralized Optimization}\label{sec:alg}
\input{icassp_new_optimization}
\section{Numerical Simulations}\label{sec:sim}
\input{icassp_simulation}
\section{Conclusion}\label{conclusion_sec}
We considered the problem of decentralized learning over time-varying directed graphs where, due to communication constraints,
nodes communicate sparsified messages. We proposed a communication-efficient algorithm that achieves 
$\mathcal{O}(\frac{\mathrm{ln}T}{\sqrt{T}})$ convergence rate for general decentralized convex optimization tasks. As part of the future work, 
it is of interest to study reduction of the computational cost of the optimization procedure by extending the results to the setting where
network agents rely on stochastic gradients.

\appendix
\input{icassp_appendix}
	\bibliography{bib}
	\bibliographystyle{acm}

\end{document}

%% file: defpack.tex
\usepackage{microtype}
\usepackage{graphicx}
\usepackage{subcaption}
\usepackage{booktabs} 

\usepackage{amsthm}
\usepackage{amsmath,amssymb,amsfonts}
\usepackage{algorithm}
\usepackage{algorithmic}
\usepackage{enumerate}

\newtheorem{definition}{Definition}[]
\newtheorem{proposition}{Proposition}
\newtheorem{assumption}{Assumption}[]

\newtheorem{theorem}{Theorem}[]

\newtheorem{lemma}[]{Lemma}

\def\x{{\mathbf x}}
\def\z{{\mathbf z}}

\def\y{{\mathbf y}}

\def\G{{\mathcal G}}
\def\E{{\mathcal E}}

\def\R{{\mathcal{R}}}

\newcommand{\ts}{\textsuperscript}

%% file: icassp_intro.tex
In recent years, decentralized optimization has attracted considerable interest from the machine learning, 
signal processing, and control communities 
\cite{ren2005consensus,ren2007information,nedic2014distributed,mcmahan2017communication}. 
We consider the setting where a collection of agents attempts to minimize an objective that consists of 
functions distributed among the agents; each agent evaluates one of the functions on its local data.
Formally, this optimization task can be stated as
\begin{equation}\label{eq:prob}
\min_{\mathbf{x} \in \mathbb{R}^d} \left[f(\mathbf{x}):=\frac{1}{n}\sum_{i=1}^n f_i(\mathbf{x})\right],
\end{equation}
where $n$ is the number of agents and $f_i: \mathbb{R}^d \to \mathbb{R}$ is the function assigned to
the $i\ts{th}$ node, $i \in [n]:=\left\{1, ..., n \right\}$. The agents collaborate by exchanging information
over a network modeled by a time-varying directed graph  $\G(t)=(|n|, \E(t))$, where $\E(t)$ denotes the
set of edges at time $t$; agent $i$ can send a message to agent $j$ at time $t$ if there exist an
edge from $i$ to $j$ at $t$, i.e., if $\left\{i, j \right\} \in \E(t)$.

The described setting has been a subject of extensive studies over the last decade, leading to a number
of seminal results \cite{nedic2009distributed, johansson2010randomized,wei2012distributed,duchi2011dual, nedic2015decentralized,stich2018sparsified,koloskova2019decentralized}. Majority of prior work assumes
symmetry in the agents' communication capabilities, i.e., models the problem using undirected graphs. 
However, the assumption of symmetry is often violated and the graph that captures properties of the 
communication network should be directed. Providing provably convergent decentralized convex optimization 
schemes over directed graphs 
is challenging; technically, this stems from the fact that unlike in undirected graphs, the so-called mixing 
matrix of a directed graph is not doubly stochastic. The existing prior work in the directed graph settings
includes the grad-push algorithm \cite{kempe2003gossip,nedic2014distributed}, which compensates for the 
imbalance in a column-stochastic mixing matrix by relying on local normalization scalars, and the directed 
distributed gradient descent (D-DGD) scheme \cite{xi2017distributed} which carefully tracks link changes
over time and their impact on the mixing matrices. Assuming convex local function, both of these methods
achieve $\mathcal{O}(\frac{\mathrm{ln}T}{\sqrt{T}})$ convergence rate. 

In practice, communication bandwidth is often limited and thus the amount of information that can be
exchanged between the agents is restricted. This motivates design of decentralized optimization schemes
capable of operating under communication constraints; none of the aforementioned methods considers
such settings. Recently, techniques that address communication constraints in decentralized 
optimization by quantizing or sparsifying messages exchanged between participating agents have been 
proposed in literature \cite{wen2017terngrad, zhang2017zipml,stich2018sparsified}. Such schemes have 
been deployed in the context of decentralized convex optimization over undirected networks 
\cite{koloskova2019decentralized} as well as in {\it fixed} directed networks \cite{taheriquantized}.
However, there has been no prior work on communication-constrained decentralized learning over 
time-varying directed networks.

In this paper we propose, to our knowledge the first, communication-sparsifying scheme for decentralized convex 
optimization over {\it directed} networks, and provide formal guarantees of its convergence; in particular, 
we show that the proposed method achieves $\mathcal{O}(\frac{\mathrm{ln}T}{\sqrt{T}})$ convergence rate. Experiments
demonstrate efficacy of the proposed scheme.

%% file: icassp_problem.tex
Assume that a collection of agents aims to collaboratively find the unique solution to decentralized convex optimization 
\eqref{eq:prob}; let us denote this solution by $\x^*$ and assume, for simplicity, that $\mathcal{X} = \mathbb{R}^d$. The
agents, represented by nodes of a directed time-varying graph, are allowed to exchange sparsified messages. In the
following, we do not assume smoothness or strong convexity of the objective; however, our analysis can be extended to
such settings.

Let $W_{in}^t$ (row-stochastic) and $W_{out}^t$ (column-stochastic) denote the in-neighbor and 
out-neighbor connectivity matrix at time $t$, respectively. 
Moreover, let $\mathcal{N}_{{in},i}^t$ be the set of nodes that can send information to node $i$ (including $i$), and 
$\mathcal{N}_{out,j}^t$ the set of nodes that can receive information from node $j$  (including $j$) at time $t$. We 
assume that both $\mathcal{N}_{{in},i}^t$ and $\mathcal{N}_{out,i}^t$ are known to node $i$. A simple policy for 
designing $W^t_{in}$ and $W^t_{out}$ is to set
\begin{equation}\label{eq:Ws}
[W^t_{in}]_{ij} = 1/|\mathcal{N}^t_{in, i}|, \quad [W^t_{out}]_{ij}= 1/|\mathcal{N}^t_{out, j}|.
\end{equation}
We assume that the constructed mixing matrices have non-zero spectral gaps; this is readily satisfied in a variety of
settings including when the union graph is jointly-connected. Matrices $W_{in}^t$ and $W_{out}^t$ can be used to
synthesize the mixing matrix, as formally stated in Section \ref{sec:alg} (see Definition~1). 


To reduce the size of the messages exchanged between agents in a network, we perform {\it sparsification}.
In particular, each node uniformly at random selects and 
communicates $k$ out of $d$ entries of a $d$-dimensional message. To formalize
this, we introduce a sparsification operator $Q: \mathbb{R}^d \rightarrow \mathbb{R}^d$. The operator $Q$ is
biased, i.e., $\mathbb{E}[Q(\x)] \neq \x$, and has variance that depends on the norm of its argument, 
$\mathbb{E}[\|Q(\x)-\x\|^2] \propto \|\x\|^2$. Biased compression operators have previously been considered in the
context of time-invariant networks 
\cite{stich2018sparsified,koloskova2019decentralized,koloskova2019decentralizeda,taheriquantized} but are not
encountered in time-varying network settings.

%% file: icassp_new_optimization.tex

A common strategy to solving decentralized optimization problems is to orchestrate exchange of messages between agents
such that each message consists of a combination of compressed messages from neighboring nodes and a gradient noise
term. The gradient term is rendered vanishing by adopting a decreasing stepsize scheme; this ensures that the agents in the
network reach a consensus state which is the optimal solution to the optimization problem.

To meet communication constraints, messages may be sparsified; however, simplistic introduction of sparsification to the 
existing methods, 
e.g., \cite{kempe2003gossip,cai2012average,cai2014average,nedic2014distributed}, may have adverse effect on their
convergence -- indeed, modified schemes may only converge to a neighborhood of the optimal solution or even end up 
diverging. This
is caused by the non-vanishing error due to the bias and variance of the sparsification operator. We note that the impact of sparsification on the entries of a state vector in the network can be interpreted as that of link failures; this motivates us to 
account for it in the structure of the connectivity matrices. Specifically, we split the vector-valued decentralized problem into 
$d$ individual scalar-valued sub-problems with the coordinate in-neighbor and out-neighbor connectivity matrices, 
$\{W_{in,m}^t\}_{m=1}^d$ and $\{W_{out,m}^t\}_{m=1}^d$, specified for each time $t$. If an entry is sparsified at time $t$ 
(i.e., set to zero 
and not communicated), the corresponding coordinate connectivity matrices are no longer stochastic. To handle this
issue, we {\it re-normalize} the connectivity matrices $\{W_{in,m}^t\}_{m=1}^d$ and $\{W_{out,m}^t\}_{m=1}^d$, ensuring 
their row stochasticity and column stochasticity, respectively; node $i$ performs re-normalization of the $i\ts{th}$ row of 
$\{W_{in,m}^t\}_{m=1}^d$ and $i\ts{th}$ column of $\{W_{out,m}^t\}_{m=1}^d$ locally. We denote by $\{A_{m}^t\}_{m=1}^d$ 
and $\{B_{m}^t\}_{m=1}^d$ the weight matrices resulting from the re-normalization of $\{W_{in,m}^t\}_{m=1}^d$ and 
$\{W_{out,m}^t\}_{m=1}^d$, respectively. 


Following the work of \cite{cai2012average} on average consensus, we introduce an auxiliary vector 
$\mathbf{y}_i\in \mathbb{R}^d$ for each node. Referred to as the surplus vector, $\mathbf{y}_i\in \mathbb{R}^d$ records 
variations of the state vectors over time and is used to help ensure the state vectors approach the consensus state.
At time step $t$, node $i$ compresses $\mathbf{x}_i^t$ and $\mathbf{y}_i^t$ and sends both to the current out-neighbors. 
To allow succinct expression of the update rule, we introduce $\mathbf{z}_i^t \in \mathbb{R}^d$ defined as
\begin{equation}\label{eq:zdef}
\mathbf{z}_i^t = \begin{cases}
\mathbf{x}_i^t, & i \in \left\{1, ..., n \right\} \\
\mathbf{y}_{i-n}^t, & i \in \left\{n+1, ..., 2n \right\}.
\end{cases}
\end{equation}
The sparsification operator $Q(\cdot)$ is applied to $\mathbf{z}_i^t$, resulting in $Q(\mathbf{z}_i^t)$; we denote the
$m\ts{th}$ entry of the sparsified vector by $[Q({\z}_{i}^t)]_m$. The aforementioned weight matrix $A^t_m$ is formed
as
\begin{equation}\label{eq:normA}
[A^t_m]_{ij}=\begin{cases}
\frac{[W^t_{in, m}]_{ij}}{\sum_{j\in\mathcal{S}_m^t(i,j)} [W^t_{in, m}]_{ij}} & \text{if } j\in\mathcal{S}_m^t(i,j)\\
0 & \mathrm{otherwise},
\end{cases}
\end{equation}
where $\mathcal{S}_m^t(i,j) := \{j|j\in\mathcal{N}^t_{in,i},[Q({\z}_{j}^t)]_m \neq 0\}\cup\{i\}$. Likewise, $B^t_m$ is defined 
as
\begin{equation}\label{eq:normB}
[B^t_m]_{ij}=\begin{cases}
\frac{[W^t_{out, m}]_{ij}}{\sum_{i\in\mathcal{T}_m^t(i,j)} [W^t_{out, m}]_{ij}} & \text{if } i\in\mathcal{T}_m^t(i,j)\\
0 & \mathrm{otherwise},
\end{cases}
\end{equation}
where $\mathcal{T}_m^t(i,j) := \{i|i\in\mathcal{N}^t_{out,j},[Q({\z}_{i}^t)]_m \neq 0\}\cup\{j\}$.

To obtain the update rule for the optimization algorithm, we first need to define the {\it mixing matrix} of a directed network with sparsified messages. 

\begin{definition}
At time $t$, the $m\ts{th}$ mixing matrix of a time-varying directed network deploying sparsified messages, 
$\Bar{M}_m^t \in \mathbb{R}^{2n \times 2n}$, is a matrix with eigenvalues $1=|\lambda_1(\Bar{M}_m^t)| = |\lambda_2(\Bar{M}_m^t)| \geq |\lambda_3(\Bar{M}_m^t)| \geq \cdots |\lambda_{2n}(\Bar{M}_m^t)|$ that is constructed from the current network topology as
\begin{equation}\label{matrixdefbar}
\Bar{M}_m^t=\left[\begin{matrix}
   A_m^t & \mathbf{0} \\
   I-A_m^t & B_m^t \\
  \end{matrix}
\right],
\end{equation}
where $A_m^t$ and $B_m^t$ represent the $m\ts{th}$ normalized in-neighbor and out-neighbor connectivity matrices at time $t$, respectively. 
\end{definition}

With $\mathbf{z}_i^t$ and $\bar{M}_m^t$ defined in \eqref{eq:zdef} and \eqref{matrixdefbar}, respectively, node $i$ updates
the $m\ts{th}$ component of its message according to
\begin{equation}\label{opti_update}
    \begin{aligned}
    z_{im}^{t+1} &=\sum_{j=1}^{2n}[\bar{M}^t_m]_{ij} [Q(\z_{j}^t)]_m  + \mathbbm{1}_{\left\{t\ \text{mod}\ \mathcal{B} = \mathcal{B}-1\right\}} \epsilon [F]_{ij} z_{jm}^{\mathcal{B} \lfloor t/\mathcal{B} \rfloor} \\
    & \quad - \mathbbm{1}_{\left\{t\ \text{mod}\ \mathcal{B} = \mathcal{B}-1\right\}} \alpha_{ \lfloor t/\mathcal{B} \rfloor} g_{im}^{\mathcal{B} \lfloor t/\mathcal{B} \rfloor},
    \end{aligned}
\end{equation}
where $g_{im}^t$ denotes the $m\ts{th}$ entry of the gradient vector $\mathbf{g}_i^t$ constructed as
\begin{equation}
    \mathbf{g}_i^t =\begin{cases}
\nabla f_i(\mathbf{x}_i^t), & i \in \left\{1, ..., n \right\}\\
\mathbf{0}, & i \in \left\{n+1, ..., 2n \right\}.
\end{cases}
\end{equation}
Moreover, $F =\left[ \begin{matrix} \mathbf{0} & I \\
\mathbf{0} & -I 
\end{matrix} \right]$, and $\alpha_t$ is the stepsize at time $t$. 

In \eqref{opti_update}, the update of vectors $\z_i^t$ consists of a mixture of the compressed state vectors and surplus 
vectors, and includes a vanishing gradient computed from history. The mixture of compressed messages can 
be interpreted as obtained by sparsification and multiplication with the mixing matrix from the previous time steps, 
except for the times when 
\begin{equation}\label{eq:mod}
    t \mod \mathcal{B} = \mathcal{B}-1.
\end{equation}
When $t$ satisfies \eqref{eq:mod}, the update of $\z_i^t$ incorporates stored vectors 
$\z_i^{\mathcal{B}\lfloor t/\mathcal{B} \rfloor}$. Note that $\z_i^{\mathcal{B}\lfloor t/\mathcal{B} \rfloor}$ is multiplied by 
$\epsilon F$, where the perturbation parameter $\epsilon$ determines the extent $F$ affects the update. One can show
that $\epsilon F$, in combination with the mixing matrix $\Bar{M}_m^t$, guarantees non-zero spectral gap of the product 
matrix over $\mathcal{B}$ consecutive time steps starting from $t = k\mathcal{B}$. Similarly, gradient term 
$\alpha_{ \lfloor t/\mathcal{B} \rfloor} g_{im}^{\mathcal{B} \lfloor t/\mathcal{B} \rfloor} $, computed using state vectors 
$\x_i^{t - (\mathcal{B}-1)}$, participates in the update when \eqref{eq:mod} holds. We formalize the proposed 
procedure as Algorithm \ref{alg:B}. 

{\color{purple}{
 \begin{algorithm}[t]
\caption{Communication-Sparsifying Jointly-Connected Gradient Descent}
\label{alg:B}
\begin{algorithmic}[1]
\STATE {\bfseries Input:} 
   $T$, 
   $\epsilon$, 
   $\mathbf{x}^0$, 
   $\mathbf{y}^0=\mathbf{0}$, 
   
\STATE {set $\mathbf{z}^0=[\mathbf{x}^0; \mathbf{y}^0]$} 
\FOR{each $t \in [0, 1,..., T]$}
\STATE generate non-negative matrices $W^t_{in}$, $W^t_{out}$

\FOR{each $m \in [1, ..., d]$}
\STATE construct a row-stochastic $A^t_m$  and a column-stochastic $B^t_m$ according to \eqref{eq:normA} and \eqref{eq:normB}

\STATE construct $\bar{M}^t_m$ according to \eqref{matrixdefbar}
\FOR{each $i \in \left\{1, ..., 2n \right\}$}
\STATE Update $z_{im}^{t+1}$ according to \eqref{opti_update}
\ENDFOR
\ENDFOR
\ENDFOR

\end{algorithmic}
\end{algorithm}
}
}

{\bf Remark.} It is worth pointing out that in Algorithm \ref{alg:B} each node needs to store local messages of size $4d$
(four $d$-dimensional vectors: the current state and surplus vectors, past surplus vector, and local gradient 
vector). Only the two current vectors may be communicated to the neighboring nodes while the other two 
vectors are used locally when \eqref{eq:mod} holds. Note that $\bar{M}_m^t$ has column sum equal to one but it is not 
column-stochastic due to having negative entries. Finally, note that when $\mathcal{B} = 1$, the network is strongly 
connected at all times. 

\begin{figure*}[htbp]
	\centering
	\minipage[t]{1\linewidth}
	\begin{subfigure}[t]{.245\linewidth}
		\includegraphics[width=\textwidth]{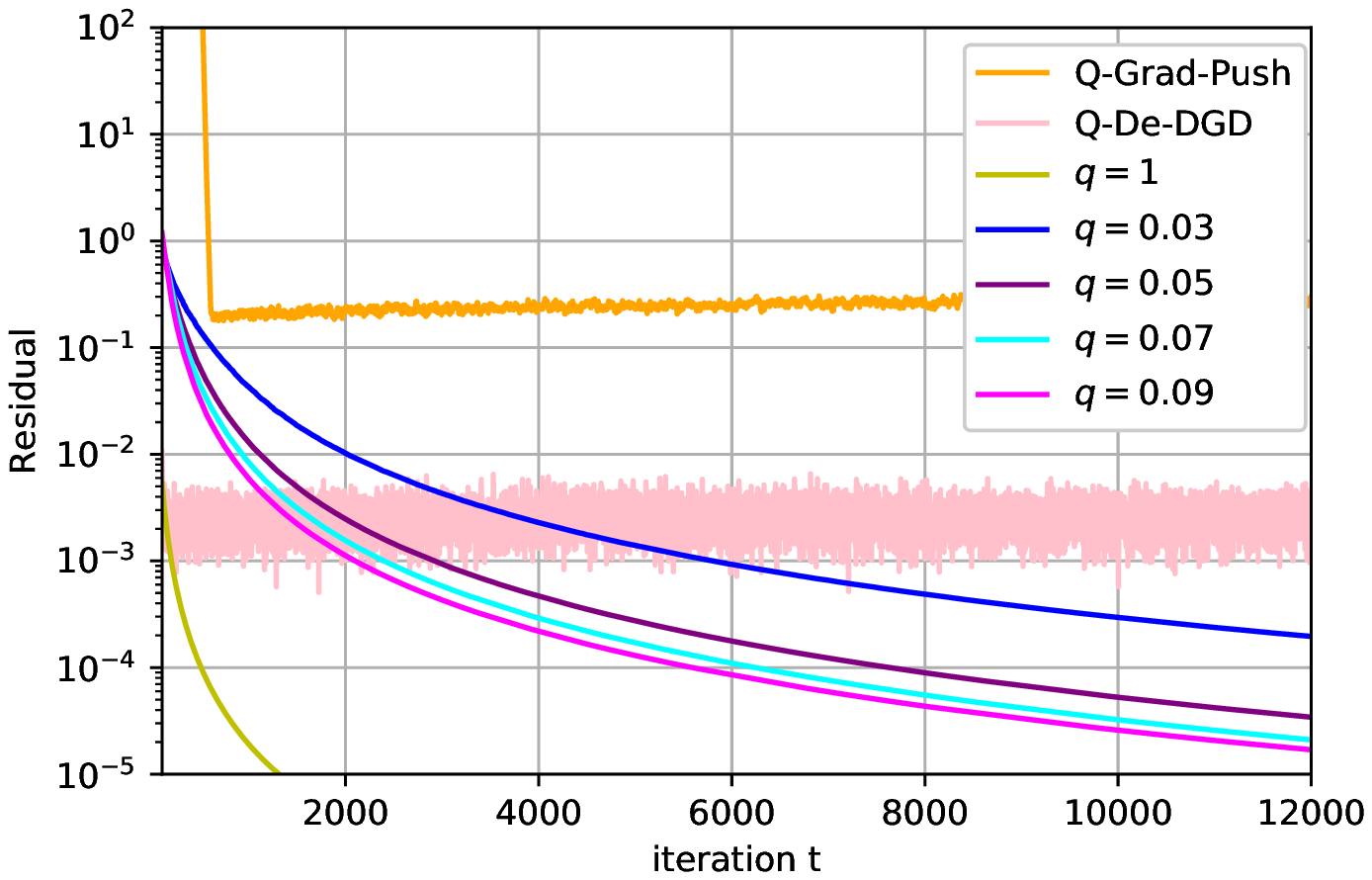}
		\caption{\footnotesize  Residual: $\mathcal{B} = 1$}
	\end{subfigure}
	\begin{subfigure}[t]{.245\linewidth}
		\includegraphics[width=\linewidth]{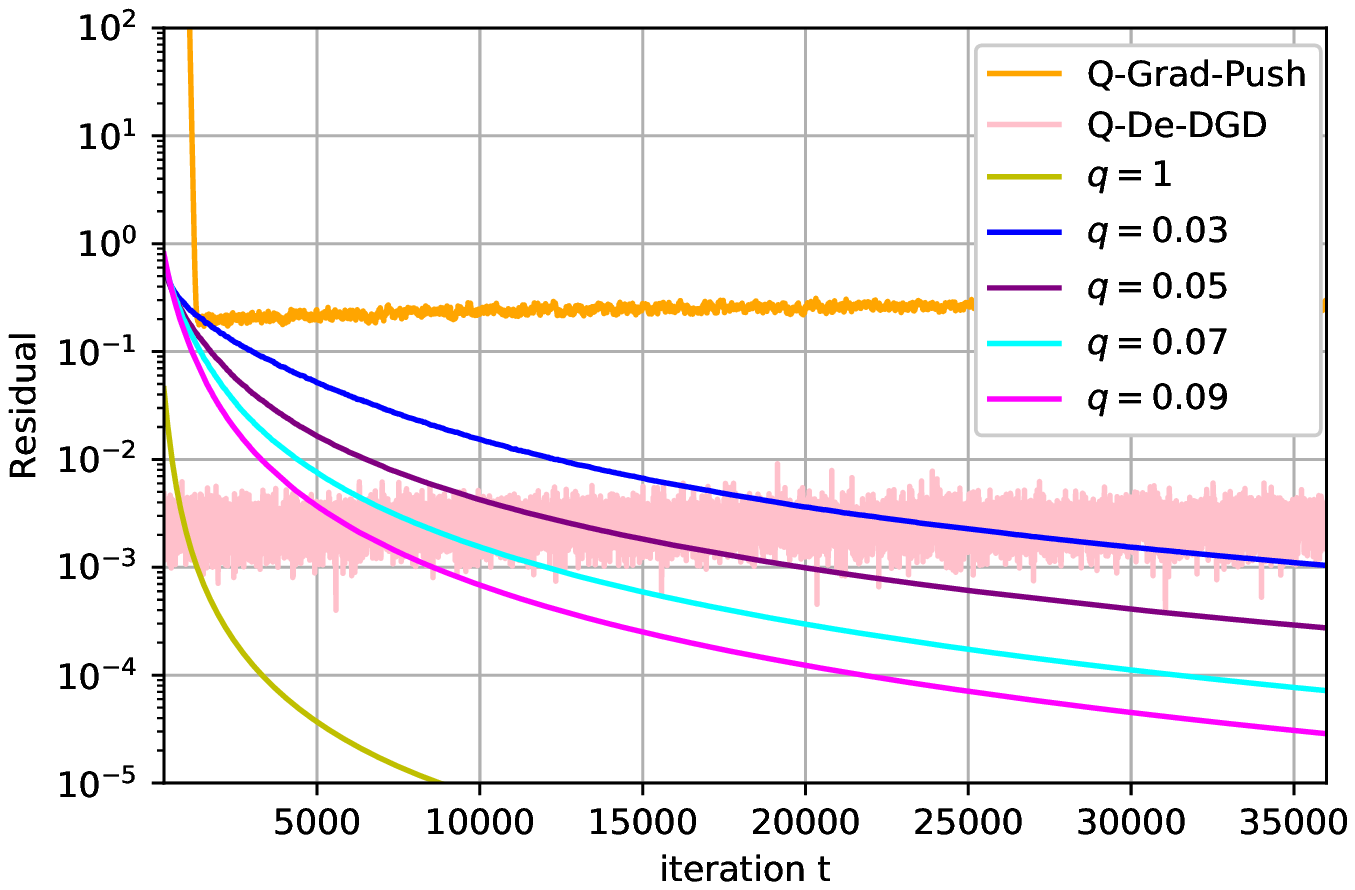}
		\caption{\footnotesize Residual: $\mathcal{B} = 3$}
	\end{subfigure}
		\begin{subfigure}[t]{.245\linewidth}
		\includegraphics[width=\textwidth]{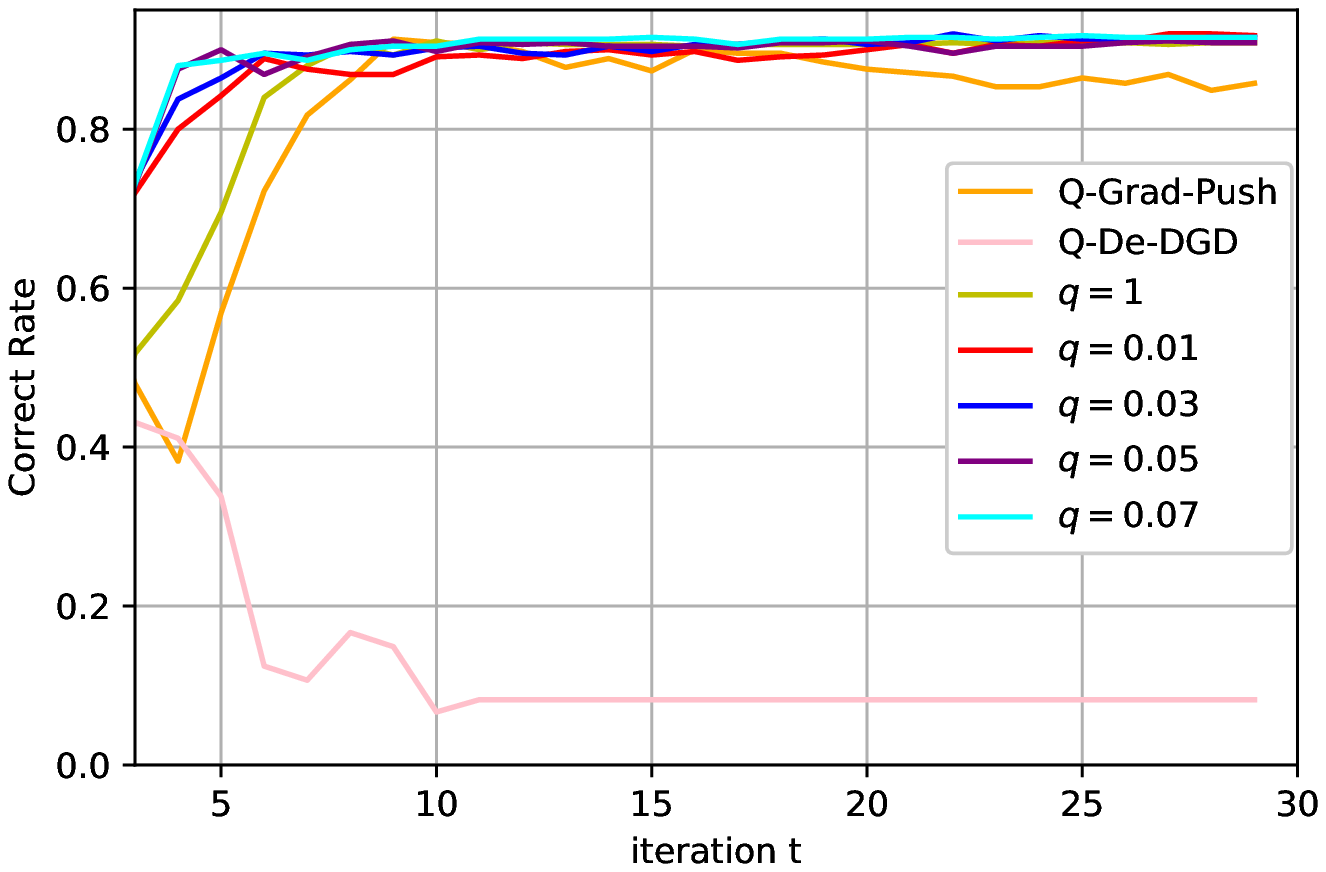}
		\caption{\footnotesize  Correct rate: $\mathcal{B} = 1$}
	\end{subfigure}
	\begin{subfigure}[t]{.245\linewidth}
		\includegraphics[width=\linewidth]{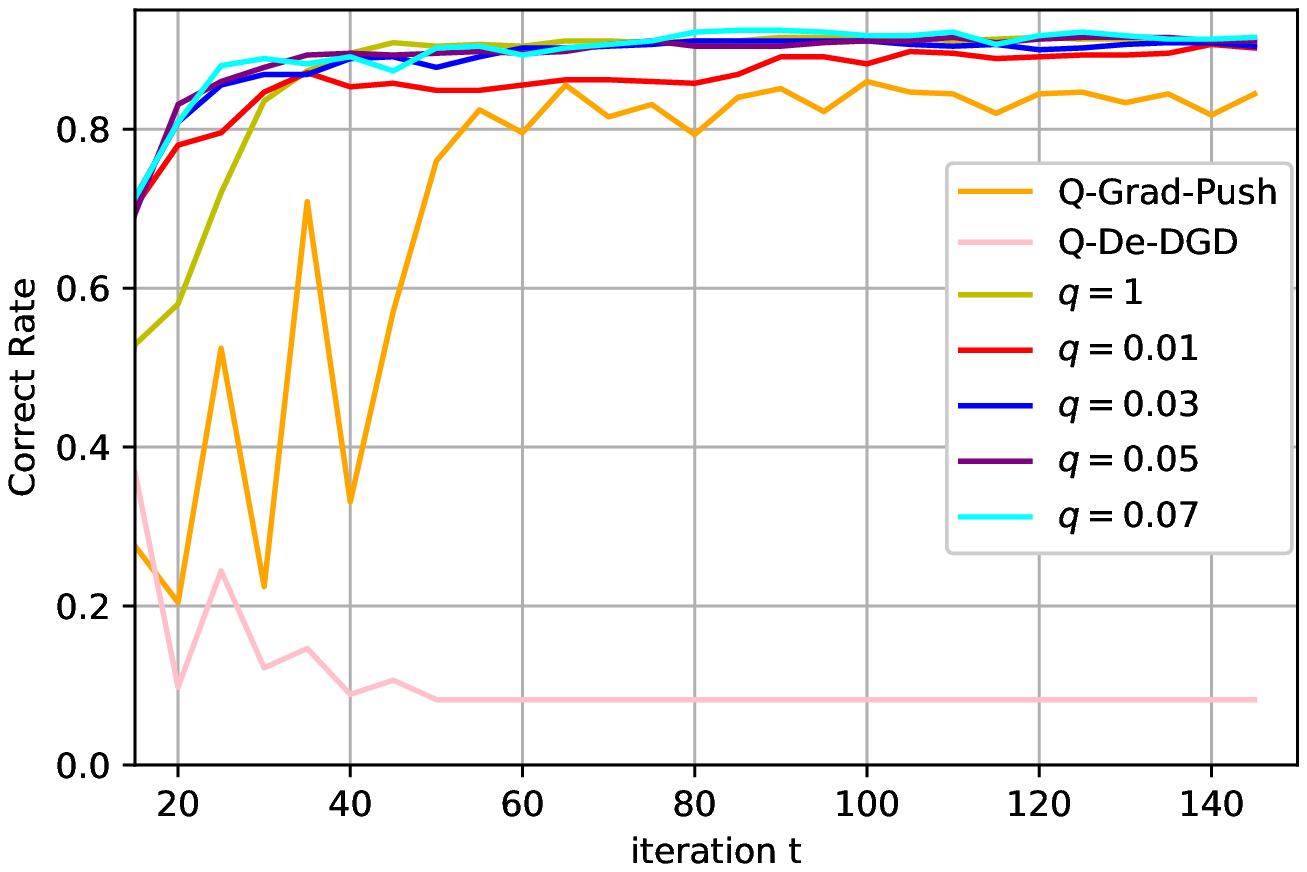}
		\caption{\footnotesize Correct rate: $\mathcal{B} = 5$}
	\end{subfigure}
	\caption{Linear regression on a jointly connected network with $\mathcal{B} = 1, 3, \epsilon = 0.05$, see (a), (b); logistic regression on a jointly connected network with $\mathcal{B} = 1, 5, \epsilon = 0.01$, see (c), (d).
	}
	\label{fig:linear_regression}
	\endminipage 
		\vspace{-0.4cm}
\end{figure*}

\subsection{Convergence Analysis}

Let $\bar{M}_m(T:s) = \bar{M}_m^{T} \bar{M}_m^{T-1} \cdots \bar{M}_m^s$ denote the product of a sequence of consecutive 
mixing matrices from time $s$ to $T$, with the superscript indicating the time and the subscript indicating the entry position. 
The perturbed product, $M_m((k+1)\mathcal{B}-1:k\mathcal{B}) $, is obtained from adding the perturbation term $\epsilon F$ 
to the product of mixing matrices as
\begin{equation}\label{mproduct}
M_m((k+1)\mathcal{B}-1:k\mathcal{B}) = \bar{M}_m((k+1)\mathcal{B}-1:k\mathcal{B}) + \epsilon F. \end{equation}
To proceed, we require the following assumptions.
\begin{assumption}\label{assumptionAll}
The mixing matrices, stepsizes, and the local objectives satisfy:
\begin{enumerate}[(i)]
\item $\forall k \geq 0, 1 \leq m \leq d$, there exists some $0 < \epsilon_0 < 1$ such that the perturbed product, $M_m((k+1)\mathcal{B}-1:k\mathcal{B})$ has a non-zero spectral gap $\forall \epsilon$ such that $0 < \epsilon < \epsilon_0$.
\item  For a fixed $\epsilon \in (0,1)$, the set of all possible mixing matrices $\{\bar{M}_m^t\}$ is a finite set.
\item The sequence of stepsizes, $\left\{\alpha_t\right\}$, is non-negative and satisfies $\sum_{t=0}^{\infty}\alpha_t=\infty, \   \sum_{t=0}^{\infty}\alpha_t^2<\infty$.
\item $\forall 1 \leq i \leq n, 1 \leq m \leq d, t \geq 0$, there exists some $D>0$ such that $|g_{im}^t|<D$.
\end{enumerate}
\end{assumption}

Given the weight matrices scheme in \eqref{eq:Ws}, assumptions (i) and (ii) hold for a variety of network structures. 
Assumptions (iii) and (iv) are common in decentralized optimization \cite{nedic2009distributed,nedic2014distributed,xi2017distributed}) and help guide nodes in the network to a consensus 
that approaches the global optimal solution. We formalize our main theoretical results in Theorem~\ref{theorem3}, which 
establishes convergence of Algorithm \ref{alg:B} to the optimal solution. The details of the proof of the theorem is stated in Section A of the appendix.


\begin{theorem}\label{theorem3}
Suppose Assumption \ref{assumptionAll} holds. Let $\x^*$ be the unique optimal solution and $f^* = f(\x^*)$. Then
\begin{equation}
    \begin{aligned}
    2\sum_{k=0}^{\infty}\alpha_k (f(\Bar{\mathbf{z}}^{k\mathcal{B}})-f^*) & \leq n\|\Bar{\mathbf{z}}^0 -\mathbf{x}^*\| + nD'^2\sum_{k=0}^{\infty}\alpha_k^2\\&\quad +\frac{4D'}{n}\sum_{i=1}^n \sum_{k=0}^{\infty} \alpha_k \| \mathbf{z}_i^{k\mathcal{B}}-\Bar{\mathbf{z}}^{k\mathcal{B}}\|,
\end{aligned}  
\end{equation}
where $D' = \sqrt{d}D$ and $\bar{\z}^t = \frac{1}{n}\sum_{i=1}^n \x_i^t + \frac{1}{n}\sum_{i=1}^n \y_i^t$. 
\end{theorem}
Note that since $\sum_{t=0}^{\infty}\alpha_t=\infty$, it is straightforward to see that  Theorem  \ref{theorem3} implies $\lim_{t \to \infty}f(\mathbf{z}_i^t)=f^*$
for every agent $i$, thereby establishing convergence of Algorithm \ref{alg:B} to the global minimum of \eqref{eq:prob}. Additionally, for the stepsize $\alpha_t = \mathcal{O}(1/\sqrt{t})$, Algorithm \ref{alg:B} attains the convergence rate $\mathcal{O}(\frac{\ln T}{\sqrt{T}})$. 

%% file: icassp_simulation.tex
We test Algorithm \ref{alg:B} in applications to linear and logistic regression, and compare the results to Q-Grad-Push, obtained by 
applying simple quantization to the push-sum scheme \cite{nedic2014distributed}, and Q-De-DGD \cite{taheriquantized}. Neither of
these two schemes was developed with communication-constrained optimization over time-varying directed networks in mind --
the former was originally proposed for unconstrainted communication, while the latter is concerned with static networks. However,
since there is no prior work on decentralized optimization over time-varying directed networks under communication constraints, we 
adopt them for the purpose of benchmarking. 

We use Erdős–Rényi model to generate strongly connected instances of a graph with $10$ nodes and edge appearance probability 
$0.9$. Two uni-directional edges are dropped randomly from each such graph while still preserving strong connectivity. 
We then remove in-going and out-going edges of randomly selected nodes to create a scenario where an almost-surely strongly
connected network is formed only after taking a union of graphs over $\mathcal{B}$ time instances (see Assumption \ref{assumptionAll}). 
Finally, recall that $q$ denotes the fraction of entries that nodes communicate to their neighbors (small $q$ implies high 
compression). 

{\bf Decentralized linear regression.} First, consider the optimization problem 
$
\min_{\x} \frac{1}{n}\sum_{i=1}^n \|\mathbf{y}_i-D_i \mathbf{x}_i\|^2,
$
where $D_i \in \mathbb{R}^{200 \times 128}$ is a local data matrix with $200$ data points of size $d=128$ at node $i$, and $\mathbf{y}_i \in \mathbb{R}^{200}$ represents the local measurement vector at node $i$. We generate $\mathbf{x}^*$ from a normal distribution, and set up
the measurement model as $\mathbf{y}_i=M_i \mathbf{x}^*+\eta_i$, where $M_i$ is randomly generated from the standard normal distribution; 
finally, the rows of the data matrix are normalized to sum to one. The local additive noise $\eta_i$ is generated from a zero-mean Gaussian 
distribution with variance $0.01$. In Algorithm \ref{alg:B} and Q-Grad-Push, local vectors are initialized randomly to $\mathbf{x}_i^0$; 
Q-De-DGD is initialized with an all-zero vector. The quantization level of the benchmarking algorithms is selected to ensure that the number 
of bits those algorithms communicate is equal to that of Algorithm \ref{alg:B} when $q = 0.09$. All algorithms are run with stepsize $\alpha_t=\frac{0.2}{t}$. The performance of different schemes is quantified by the residuals $\frac{\|\mathbf{x}^t-\Bar{\mathbf{x}}\|}{\|\mathbf{x}^0-\Bar{\mathbf{x}}\|}$. 
The results are shown in Fig.\ref{fig:linear_regression} (a), (b). As shown in the subplots, for all the considered sparsification rates 
Algorithm \ref{alg:B} converges with rate proportional to $q$, while the benchmarking algorithms do not converge to 
the optimal solution.

{\bf Decentralized logistic regression.} Next, we consider a multi-class classification task on the MNIST dataset \cite{lecun1998gradient}. The 
logistic regression problem is formulated as
$$
\min_{\x} \left\{ \frac{\mu}{2}\|\mathbf{x}\|^2+\sum_{i=1}^n \sum_{j=1}^N \mathrm{ln}(1+\mathrm{exp}(-(\mathbf{m}_{ij}^T\mathbf{x}_i)y_{ij})) \right\}.
$$
The data is distributed across the network such that each node $i$ has access to $N=120$ training samples 
$(\mathbf{m}_{ij}, y_{ij}) \in \mathbb{R}^{64}\times \{0, \cdots, 9 \} $, where $\mathbf{m}_{ij}$ denotes a vectorized image with 
size $d = 64$ and $y_{ij}$ denotes the corresponding digit label. Performance of Algorithm \ref{alg:B} is again compared with 
Q-Grad-Push and Q-De-DGD; all algorithms are initialized with zero vectors. The quantization 
level of the benchmarking algorithms is selected such that the number of bits they communicate is equal to that of Algorithm 
\ref{alg:B} for $q = 0.07$. The experiment is run using the stepsize $\alpha_t=\frac{0.02}{t}$; we set $\mu=10^{-5}$. 
Fig. \ref{fig:linear_regression}  (c), (d) show the classification correct rate of Algorithm \ref{alg:B} for different sparsification and 
connectivity levels. As can be seen there, all sparsified schemes achieve the same level of the classification correct rate. The 
schemes communicating fewer messages in less connected networks converge slower, while the two benchmarking algorithms 
converge only to a neighborhood of the optimal solution.

%% file: icassp_appendix.tex
\section*{Appendices}
The appendix presents the analysis of Theorem~1 and derives the convergence rate.

\section{Elaborating on Assumption \ref{assumptionAll}(i)}

Analysis of the algorithms presented in the paper is predicated on the property of the product of consecutive mixing matrices of general time-varying graphs stated in Assumption \ref{assumptionAll}. Here we establish conditions under which this property holds for a specific graph structure, i.e., identify $\epsilon_0$ in Assumption \ref{assumptionAll} for the graphs that are jointly connected over $\mathcal{B}$ consecutive time steps. Note that when $\mathcal{B}=1$, such graphs reduce to the special case of graphs that are strongly connected at each time step. For convenience, we formally state the $\mathcal{B} > 1$ and $\mathcal{B}=1$ settings as Assumption \ref{assumption1} and \ref{assumption1'}, respectively.




\begin{assumption}\label{assumption1}
The graphs $\G_m(t)=(|n|, \E_m(t))$,
modeling network connectivity for the $m\ts{th}$ entry of the sparsified parameter vectors, are $\mathcal{B}$-jointly-connected. 
\end{assumption}

Assumption~\ref{assumption1}
implies that starting from any time step $t = k\mathcal{B}$, $k \in \mathcal{N}$, the union graph over $\mathcal{B}$ consecutive time steps is a strongly connected graph. This is a weaker requirement then the standard assumption (Assumption~\ref{assumption1'} given below) often encountered in literature on convergence analysis of algorithms for distributed optimization and consensus problems.

\begin{assumption}\label{assumption1'}
The graphs $\G_m(t)=(|n|, \E_m(t))$,
modeling network connectivity for the $m\ts{th}$ entry of the sparsified parameter vectors, are strongly connected at any time $t$. 
\end{assumption}

 Next, we state a lemma adopted from \cite{cai2012average} which helps establish that under Assumptions \ref{assumptionAll}(ii) and \ref{assumption1'}, the so-called spectral gap of the product of mixing matrices taken over a number of consecutive time steps is non-zero. 
 
\begin{lemma}\label{lemma1}
\cite{cai2012average} Suppose Assumptions \ref{assumptionAll}(ii) and \ref{assumption1'} hold. Let $\Bar{M}_m^t$ be the mixing matrix in \eqref{matrixdefbar} and $M_m^t = \Bar{M}_m^t + \epsilon F $ such that $\epsilon \in (0,\gamma_m)$, where $\gamma_m = \frac{1}{(20+8n)^n}(1-|\lambda_3(\Bar{M}_m^t)|)^n$. Then, the mixing matrix $M_m^t$ has a simple eigenvalue $1$ and all its other eigenvalues have magnitude smaller than $1$.
\end{lemma}
Note that Lemma \ref{lemma1} implies that 
\begin{equation}
    \label{Mlim}
    \mathrm{lim}_{k \to \infty}(M_m^t)^k=\left[\begin{matrix}
    \frac{\mathbf{1}_n \mathbf{1}^T_n}{n} & \frac{\mathbf{1}_n \mathbf{1}^T_n}{n} \\
    \mathbf{0} & \mathbf{0} \\
  \end{matrix}
\right],
\end{equation}
where the rate of convergence is geometric and determined by the nonzero spectral gap of the perturbed mixing matrix $M_m^t$, i.e., $1-|\lambda_2(M_m^t)|$. This implies that the local estimate, $x^t_{im}$, approaches the average consensus value $\Bar{z}^t_{m}=\frac{1}{n}\sum_{i=1}^n x^t_{im}+\frac{1}{n}\sum_{i=1}^n y^t_{im}$ while the auxiliary variable $y^t_{im}$ vanishes to 0.
 
We now utilize the insight of (\ref{Mlim}) to establish a
result that will facilitate convergence analysis of the setting described by Assumption~\ref{assumption1}. In particular, we consider the setting where in any time window of size $\mathcal{B}$ starting from time $t = k\mathcal{B}$ for some integer $k$, the union of the associated directed graphs is strongly connected.
 The following lemma will help establish that if a small perturbation, 
$\epsilon F$
 is added to the product of mixing matrices  $\bar{M}_m((k+1)\mathcal{B}-1:k\mathcal{B} ) $ 
 , then
 the product $M_m((k+1)\mathcal{B}-1:k\mathcal{B} )$ has only a simple eigenvalue $1$ while all its other eigenvalues have moduli smaller than one.

\begin{lemma}\label{lemma2}
Suppose that Assumptions \ref{assumptionAll}(ii) and \ref{assumption1} hold. 
If the parameter $\epsilon \in (0, \bar{\epsilon})$ and $\bar{\epsilon} = \min_m \gamma_m$, where $\gamma_{m} = \min_k \frac{1}{(20+8n)^n}(1-|\lambda_3(\bar{M}_m((k+1)\mathcal{B}-1:k\mathcal{B}))|)^n$, then for each $m$ the mixing matrix product $M_m((k+1)\mathcal{B}-1:k\mathcal{B} ) $ has simple eigenvalue $1$ for all integer $k \geq 0$ and $\epsilon \in (0,  \bar{\epsilon})$.
\end{lemma}
 
\begin{proof}
Consider a fixed realization of $\mathcal{B}$-strongly-connected graph sequences,\newline $\left\{\G(0), \cdots, \G(\mathcal{B}-1) \right\}.$
For $s \in 0, \cdots, \mathcal{B}-1$, $\bar{M}^s$ is block (lower) triangular and the spectrum is determined by the spectrum of the $(1, 1)$-block and the $(2, 2)$-block. Furthermore, for such $s$, $A^s$ (row-stochastic) and $B^s$ (column-stochastic) matrices have non-negative entries. Owing to the fact that the union graph over $\mathcal{B}$ iterations is strongly connected, $\Pi_{s=0}^{\mathcal{B}-1 }A^s = A^{\mathcal{B}-1}\cdots A^0 $ and $\Pi_{s=0}^{\mathcal{B}-1 }B^s = B^{\mathcal{B}-1}\cdots B^0 $ are both irreducible. Thus, $\Pi_{s=0}^{\mathcal{B}-1 }A^s $ and $\Pi_{s=0}^{\mathcal{B}-1 }B^s$ both have simple eigenvalue $1$. Recall that $\Pi_{s=0}^{\mathcal{B}-1 }\bar{M}^s$ has column sum equal to $1$, and thus we can verify that $\mathrm{rank}(\Pi_{s=0}^{\mathcal{B}-1 }\bar{M}^s - I)=2n-2$; therefore, the eigenvalue $1$ is semi-simple.
 \\~\\
Next, we characterize the change of the semi-simple eigenvalue $\lambda_1 = \lambda_2 = 1$ of $\Pi_{s=0}^{\mathcal{B}-1 }\bar{M}^s$ when a small perturbation $\epsilon F$ is added. Consider the eigenvalues of the perturbed matrix product,  $\lambda_1(\epsilon)$, $\lambda_2(\epsilon) $, which corresponds to $ \lambda_1$, $\lambda_2$, respectively. 
For all $s \in 0, \cdots, \mathcal{B}-1$, $ \bar{M}^s$ has two common right eigenvectors and left eigenvectors for eigenvalue $1$; they are the right eigenvectors and left eigenvectors of the matrix product. The right eigenvectors $y_1$, $y_2$ and left eigenvectors $z_1$, $z_1$ of the semi-simple eigenvalue $1$ are
 \begin{equation}
     Y := \begin{matrix}[y_1 & y_2]\end{matrix} = \begin{bmatrix} 0 & 1 \\ v_2 & -nv_2 \end{bmatrix}, \quad\hspace{-2mm} Z := \begin{bmatrix}z_1' \\ z_2'\end{bmatrix} = \begin{bmatrix} 1' & 1' \\ v_1' & 0 \end{bmatrix}.
 \end{equation}
 By following exactly the same steps and using Proposition $1$ in \cite{cai2012average}, we can show that for small $\epsilon >0$, the perturbed matrix product has a simple eigenvalue $1$. Further, it can be guaranteed that for $\epsilon < \frac{1}{(20+8n)^n}(1-|\lambda_3(\bar{M}(\mathcal{B}-1:0)|)^n$, the perturbed matrix product $M(\mathcal{B}-1:0 ) $ has simple eigenvalue $1$. 
 \\~\\
From Assumption \ref{assumptionAll}(ii), there is only a finite number of possible mixing matrices $\left\{\bar{M}^t_m \right\}$; if we let $\gamma_{m} = \min_k \frac{1}{(20+8n)^n}(1-|\lambda_3(\bar{M}_m((k+1)\mathcal{B}-1:k\mathcal{B}))|)^n$, starting from any time step $t = k\mathcal{B}$, the perturbed mixing matrix product $M_m(t+\mathcal{B}-1:t ) $ has simple eigenvalue $1$ for all $\epsilon < \bar{\epsilon} = \min_m \gamma_m$.
 \end{proof}
	
\section{Lemma \ref{lemma0} and its proof}

The following lemma implies that the product of mixing matrices converges to its limit 
at a geometric rate; this intermediate result can be used to establish the convergence rate of the proposed optimization Algorithm.

\begin{lemma}\label{lemma0}
Suppose that Assumptions \ref{assumptionAll}(i) and \ref{assumptionAll}(ii) hold. 
There exists $\epsilon_0 >0$ such that if $\epsilon \in (0, \epsilon_0)$
then for all $m = 1, \cdots, d$ the following statements hold.
 
\begin{enumerate}[(a)]
\item The spectral norm of 
$M_m((k+1)\mathcal{B}-1:k\mathcal{B})$ satisfies
\begin{equation}
\begin{aligned}
\rho(M_m((k+1)\mathcal{B}-1:k\mathcal{B})) - \frac{1}{n}[\mathbf{1}^T \ \mathbf{0}^T]^T[\mathbf{1}^T \  \mathbf{1}^T])  = \sigma <1
     \end{aligned}
     \end{equation}
     $\forall k \in \mathcal{N}$.
\item 
There exists $\Gamma = \sqrt{2nd}>0$ such that
\begin{equation}
    \|
    M_m(n\mathcal{B}-1:0)-\frac{1}{n}[\mathbf{1}^T \ \mathbf{0}^T]^T[\mathbf{1}^T \  \mathbf{1}^T]\|_{\infty} \leq \Gamma \sigma^{n}.
    \end{equation}
 \end{enumerate}
 \end{lemma}


 

We start this proof by introducing two intermediate lemmas: Lemma \ref{lemma11} and Lemma \ref{lemma12}. 
	
\begin{adjustwidth}{0.25cm}{0cm}\qquad
\begin{lemma}\label{lemma11}
Assume that $M_m((s+1)\mathcal{B}-1:s\mathcal{B})$ has non-zero spectral gap for each $m$.
Then the following statements hold.
\begin{enumerate}[(a)]
\item The sequence of matrix products $M_m((s+1)\mathcal{B}-1:s\mathcal{B})$ converges to the limit matrix
\begin{equation}
\mathrm{lim}_{t \to \infty} (M_m((s+1)\mathcal{B}-1:s\mathcal{B}))^t=
\left[\begin{matrix}
\frac{\mathbf{1}_n \mathbf{1}^T_n}{n} & \frac{\mathbf{1}_n \mathbf{1}^T_n}{n} \\
\mathbf{0} & \mathbf{0} \\
\end{matrix}
\right].
\end{equation}
\item Let $1=|\lambda_1(M_m((s+1)\mathcal{B}-1:s\mathcal{B}))|>|\lambda_2(M_m((s+1)\mathcal{B}-1:s\mathcal{B}))|\geq \cdots \geq |\lambda_{2n}(M_m((s+1)\mathcal{B}-1:s\mathcal{B}))|$ be the eigenvalues of $M_m((s+1)\mathcal{B}-1:s\mathcal{B})$, and let $\sigma_m=|\lambda_2(M_m((s+1)\mathcal{B}-1:s\mathcal{B}))|$; then there exists $\Gamma_m'>0$ such that
\begin{equation}
\begin{aligned}
\|(M_m((s+1)\mathcal{B}-1:s\mathcal{B}))^t-\mathcal{I}\|_{\infty} 
\leq \Gamma_m' \sigma_m^{t},
\end{aligned}
\end{equation}
where $\mathcal{I}:=\frac{1}{n}[\mathbf{1}^T \ \mathbf{0}^T]^T[\mathbf{1}^T \  \mathbf{1}^T]$.
\end{enumerate}
\end{lemma}
	
\begin{proof}
For each $m$, $M_m((s+1)\mathcal{B}-1:s\mathcal{B})$ has column sum equal to $1$. According to Assumption \ref{assumptionAll}, definition of the mixing matrix (\ref{matrixdefbar}), and the construction of the product (\ref{mproduct}),
$M_m((s+1)\mathcal{B}-1:s\mathcal{B})$ has a simple eigenvalue $1$ with the corresponding left eigenvector $[\mathbf{1}^T \ \mathbf{1}^T]$ and right eigenvector $[\mathbf{1}^T \ \mathbf{0}^T]^T$.
Following Jordan matrix decomposition for the simple eigenvalue, there exist some $P, Q \in \mathcal{R}^{(2n-1)\times (2n-1)}$ such that
	\begin{equation}
	\begin{aligned}
	(M_m((s+1)\mathcal{B}-1:s\mathcal{B}))^t & =\mathcal{I}^t+P J_m^t Q
	=\mathcal{I}+P J_m^t Q. 
	\end{aligned}
\end{equation}
Let $\gamma_m$ be the second largest eigenvalue magnitude of $M_m((s+1)\mathcal{B}-1:s\mathcal{B})$; then, $\gamma_m$ is also the spectral norm of $J_m$. The proof of part (a) follows by noting that 
$
\mathrm{lim}_{t \to \infty}J_m^t = \mathbf{0}.
$
Since $\|P \|$, $\| Q \|$ and $\|J_m \|$ are finite, there exists some $\Gamma_m' >0$ such that
\begin{equation}
\begin{aligned}
& \quad \|(M_m((s+1)\mathcal{B}-1:s\mathcal{B}))^t-\mathcal{I} \|_{\infty} 
\leq \|P J_m^t Q \|_{\infty} 
\leq \Gamma_m' \sigma_m^t
\end{aligned}
\end{equation}
which completes the proof of part (b).
\end{proof}

\begin{lemma}\label{lemma12}
Suppose that for each $m$, $M_m((s+1)\mathcal{B}-1:s\mathcal{B})$ has non-zero spectral gap. Let $\sigma = \max_m \sigma_m$, where $\sigma_m$ is as defined in Lemma \ref{lemma11}.
Then, for each $m$ it holds that
\begin{equation}
\rho (M_m(T\mathcal{B}-1:0)-\frac{1}{n}[\mathbf{1}^T \ \mathbf{0}^T]^T[\mathbf{1}^T \  \mathbf{1}^T]) \leq  \sigma^{T}.
\end{equation}
\end{lemma}
\begin{proof}
We prove this lemma by induction. \\~\\
\textbf{Base step.} $T=1$. According to the selection rule of $M_m(\mathcal{B}-1:0)$ and definition of $\sigma$, the statement holds. \\~\\
\textbf{Inductive step.} Suppose for all $T_1 < T$ the statement holds. Let $T_1=T$. Since for each time step $t = k\mathcal{B}$,  $M_m(t\mathcal{B}-1:(t-1)\mathcal{B})$ has column sum equal to $1$ and has a simple eigenvalue $1$ with the corresponding left eigenvector $[\mathbf{1}^T \ \mathbf{1}^T]$ and right eigenvector $[\mathbf{1}^T \ \mathbf{0}^T]^T$, then 
\begin{align*}
 \quad M_m(T\mathcal{B}-1:0)-\mathcal{I} & =M_m(T\mathcal{B}-1:0) -\mathcal{I}M_m(\mathcal{B}-1:0) \\
 & = (M_m(T\mathcal{B}-1:\mathcal{B})-\mathcal{I} )M_m(\mathcal{B}-1:0). 
	\end{align*}
Taking the spectral norm over both hand sides after recursion and applying Gelfand corollaries, we complete the proof. \end{proof}
\end{adjustwidth}

We now continue with the proof of Lemma \ref{lemma0}.
Lemma \ref{lemma12} implies the result in part (a) of Lemma \ref{lemma0}. Due to the equivalence of matrix norms, we can obtain the desired results in Lemma \ref{lemma0} (b). In particular, for matrix $A\in \mathcal{R}^{m \times n}$ it holds that \begin{equation*}
    \frac{1}{\sqrt{n}} \|A \|_{\infty}\leq \|A \|_2 \leq \sqrt{m} \|A \|_{\infty}.
\end{equation*}
Since Lemma \ref{lemma12} shows that
$
\|M_m(T\mathcal{B}-1:0)-\mathcal{I} \|_2 \leq \sigma^T,
$
then there exists $\Gamma = \sqrt{2nd} >0$ such that
\begin{equation*}
    \|M_m(T\mathcal{B}-1:0)-\frac{1}{n}[\mathbf{1}^T \ \mathbf{0}^T]^T[\mathbf{1}^T \  \mathbf{1}^T] \|_{\infty} \leq \Gamma \sigma^T,
\end{equation*}
which completes the proof. 

\section{Lemma \ref{lemma3} and its proof}

In this part we state an intermediate lemma that establishes an upper bound on the disagreement term $\|\mathbf{z}_i^t-\mathbf{\Bar{z}}^t\|$.

\begin{lemma}\label{lemma3}
Let $\Gamma = \sqrt{2nd}$ and
Assumptions \ref{assumptionAll}(i), (ii), (iii) and (iiii) imply the following statements:
\begin{enumerate}[(a)]
    \item For $1 \leq i \leq n$ and $t = k\mathcal{B} - 1 + t'$, where $t' = 1, \cdots, \mathcal{B}$, it holds that
    \begin{equation}
            \begin{aligned}
    \|\mathbf{z}_i^{k\mathcal{B}}-\mathbf{\Bar{z}}^{k\mathcal{B}}\| & \leq \Gamma \sigma^{k}\sum_{j=1}^{2n}\sum_{m=1}^d |z^0_{jm}|+\sqrt{d}n\Gamma D\sum_{r=1}^{k-1}\sigma^{k-r}\alpha_{r-1}+2\sqrt{d} D\alpha_{k-1},
    \end{aligned}
    \end{equation}
    \begin{equation}
            \begin{aligned}
    \|\mathbf{z}_i^t-\mathbf{\Bar{z}}^t\| & \leq \Gamma (\sigma^{1/\mathcal{B}})^{t-(t'-1) } \sum_{j=1}^{2n}\sum_{m=1}^d |z^0_{jm}|+\sqrt{d}n\Gamma D\sum_{r=1}^{\lfloor
     t/\mathcal{B}\rfloor -1}\sigma^{\lfloor
     t/\mathcal{B}\rfloor-r}\alpha_{r-1}\\
     & +2\sqrt{d} D\alpha_{\lfloor
     t/\mathcal{B}\rfloor-1}\mathbbm{1}_{t' = 1}.
    \end{aligned}
    \end{equation}
    \item  For $1+n \leq i \leq 2n$ and $t = k\mathcal{B} - 1 + t'$, where $t' = 1, \cdots, \mathcal{B}$, it holds that
    \begin{equation}
            \begin{aligned}
    \|\mathbf{z}_i^{k\mathcal{B}}\| & \leq \Gamma \sigma^{k}\sum_{j=1}^{2n}\sum_{m=1}^d |z^0_{jm}|+\sqrt{d}n\Gamma D\sum_{r=1}^{k-1}\sigma^{k-r}\alpha_{r-1}+2\sqrt{d} D\alpha_{k-1},
    \end{aligned}
    \end{equation}
    \begin{equation}
    \begin{aligned}
    \|\mathbf{z}_i^t\| & \leq \Gamma (\sigma^{1/\mathcal{B}})^{t-(t'-1) } \sum_{j=1}^{2n}\sum_{m=1}^d |z^0_{jm}| 
     +\sqrt{d}n\Gamma D\sum_{r=1}^{\lfloor
     t/\mathcal{B}\rfloor -1}\sigma^{\lfloor
     t/\mathcal{B}\rfloor-r}\alpha_{r-1} \\
     & +2\sqrt{d} D\alpha_{\lfloor
     t/\mathcal{B}\rfloor-1}\mathbbm{1}_{t' = 1}.
    \end{aligned}
    \end{equation}
\end{enumerate}

\end{lemma}
	

Consider the time step $t=k\mathcal{B}-1$ for some integer $k$ and rewrite the update \eqref{opti_update} as
\begin{equation}
\begin{aligned}
z_{im}^{t+1} &=\sum_{j=1}^{2n}[\bar{M}^t_m]_{ij} [Q(\z_{j}^t)]_m  + \mathbbm{1}_{\left\{t\mod \mathcal{B} = \mathcal{B}-1\right\}} \epsilon [F]_{ij} z_{jm}^{\lfloor t/\mathcal{B} \rfloor} - \mathbbm{1}_{\left\{t\mod \mathcal{B} = \mathcal{B}-1\right\}} \alpha_t g_{im}^{\lfloor t/\mathcal{B} \rfloor} \\
& = \sum_{j=1}^{2n}[\bar{M}^t_m]_{ij} z_{jm}^t + \mathbbm{1}_{\left\{t\mod \mathcal{B} = \mathcal{B}-1\right\}} \epsilon [F]_{ij} z_{jm}^{\lfloor t/\mathcal{B} \rfloor} - \mathbbm{1}_{\left\{t\mod \mathcal{B} = \mathcal{B}-1\right\}} \alpha_t g_{im}^{\lfloor t/\mathcal{B} \rfloor}.
\end{aligned}
\end{equation}
Establishing recursion, we obtain
\begin{equation}
\begin{aligned}
z_{im}^{k\mathcal{B}} & = 
    \sum_{j=1}^{2n}[M_m(k\mathcal{B}-1:0)]_{ij}z_{jm}^0 - \sum_{r=1}^{k-1}\sum_{j=1}^{2n}[M_m((k-1)\mathcal{B}-1:(r-1)\mathcal{B})]_{ij}\alpha_{r-1}g_{jm}^{r-1}\\ &-\alpha_{k-1}g_{im}^{(k-1)\mathcal{B}}.
\end{aligned}
\end{equation}
Using the fact that $M_m(s_2:s_1)$ has column sum equal to $1$ for all $s_2 \geq s_1 \geq 0$, we can represent $\bar{z}_m^{k\mathcal{B}}$ as
\begin{equation}
\begin{aligned}
    \bar{z}_m^{k\mathcal{B}} &= \frac{1}{n}\sum_{j=1}^{2n}z_{jm}^0-\frac{1}{n}\sum_{r=1}^{k-1}\sum_{j=1}^{2n}\alpha_{r-1} g_{jm}^{r-1}-\frac{1}{n}\sum_{j=1}^{n}\alpha_{k-1}g_{jm}^{(k-1)\mathcal{B}}.
\end{aligned}
\end{equation}
By combining the last two expressions,
\begin{equation}
\begin{aligned}
 \|\z_i^{k\mathcal{B}} - \bar{\z}^{k\mathcal{B}} \| 
& \leq \|\sum_{j=1}^{2n}([M_m(k\mathcal{B}-1:0)]_{ij}-\frac{1}{n}) z_{jm}^0 \| \\
&  + \|\sum_{r=1}^{k-1}\sum_{j=1}^{2n}([M_m((k-1)\mathcal{B}-1:(r-1)\mathcal{B})]_{ij}  -\frac{1}{n} )\alpha_{r-1}g_{jm}^{r-1} \| \\
& +  \|\alpha_{k-1}(\mathbf{g}_{i}^{(k-1)\mathcal{B}} - \frac{1}{n}
\sum_{j=1}^{n}\mathbf{g}_{j}^{(k-1)\mathcal{B}} 
) \| .
\end{aligned}
\end{equation}
The proof of part (a) is completed by summing $m$ from $1$ to $d$ and applying the results of Lemma \ref{lemma0}, recalling the fact that $\rho(\bar{M}_m(t-1:0))$ has non-zero spectral gap for all $m$, $t$, and invoking the relationship $\|\x \|_2 \leq \|\x \|_1 \leq \sqrt{d}\|\x \|_2$ for $\x \in \R^d$.
Proof of the first inequality in part (b) follows the same line of reasoning.

	
To show the correctness of the second inequality in both (a) and (b), we use the fact that for $t \mod \mathcal{B} \neq \mathcal{B}-1$, 
\begin{equation}
z_{im}^{t+1} =\sum_{j=1}^{2n}[\bar{M}^t_m]_{ij} [Q(\z_{j}^t)]_m = \sum_{j=1}^{2n}[\bar{M}^t_m]_{ij} z_{jm}^t
	\end{equation}
	and rewrite $k = \frac{t-(t'-1)}{\mathcal{B}}$. 
This concludes the proof of Lemma \ref{lemma3}.

\section{Proof of Theorem \ref{theorem3}}
	

Recall the update \eqref{opti_update}
and note that
$
\Bar{\z}^{(k+1)\mathcal{B}}=\Bar{\z}^{k\mathcal{B}}-\frac{\alpha_t}{n}\sum_{i=1}^n \nabla f_i(\z_i^{k\mathcal{B}}).
$
We thus have that
\begin{equation}
\begin{aligned}
\|\Bar{\z}^{k\mathcal{B}+k}-\x^* \|^2 = \|\frac{\alpha_k}{n}\sum_{i=1}^n \nabla f_i(\z_i^{k\mathcal{B}}) \|^2 +\|\Bar{\z}^{k\mathcal{B}}-\x^* \|^2-\frac{2\alpha_k}{n}\sum_{i=1}^n \langle \Bar{\z}^{k\mathcal{B}}-\x^*, \nabla f_i(\z_i^{k\mathcal{B}}) \rangle.
\end{aligned}
\end{equation}
On the other hand,
\begin{equation}
\begin{aligned}
\|\Bar{\z}^{t}-\x^* \|^2 
= \|\Bar{\z}^{k\mathcal{B}}-\x^* \|^2+\|\frac{\alpha_k}{n}\sum_{i=1}^n \nabla f_i(\z_i^{k\mathcal{B}}) \|^2  -\frac{2\alpha_k}{n}\sum_{i=1}^n \langle \Bar{\z}^{k\mathcal{B}}-\x^*, \nabla f_i(\z_i^{k\mathcal{B}}) \rangle
\end{aligned}
\end{equation}
for 
$t = k\mathcal{B}-1+t'$ and $t' = 1, \cdots, \mathcal{B}-1$. 
Therefore, 
\begin{equation}
	\begin{aligned}
	\|\Bar{\z}^{k\mathcal{B}+k}-\x^* \|^2 = \|\Bar{\z}^{k\mathcal{B}}-\x^* \|^2+\|\frac{\alpha_k}{n}\sum_{i=1}^n \nabla f_i(\z_i^{k\mathcal{B}}) \|^2  -\frac{2\alpha_k}{n}\sum_{i=1}^n \langle \Bar{\z}^{k\mathcal{B}}-\x^*, \nabla f_i(\z_i^{k\mathcal{B}})  \rangle.
	\end{aligned}
\end{equation}
Since $|g_{im}| \leq D$ and $D'=\sqrt{d}D$, and by invoking the convexity of $f$,
\begin{equation}
\begin{aligned}
\langle \Bar{\z}^{k\mathcal{B}}-\x^*, \nabla f_i(\z_i^{k\mathcal{B}})  \rangle  & = \langle \Bar{\z}^{k\mathcal{B}}-\z_i^{k\mathcal{B}}, \nabla f_i(\z_i^{k\mathcal{B}})  \rangle+  \langle \z_i^{k\mathcal{B}}-\x^*, \nabla f_i(\z_i^{k\mathcal{B}})  \rangle \\
	& \geq -D'\|\Bar{\z}^{k\mathcal{B}}-\z_i^{k\mathcal{B}} \|+f_i(\z_i^{k\mathcal{B}})-f_i(\Bar{\z}^{k\mathcal{B}}) +f_i(\Bar{\z}^{k\mathcal{B}})-f_i(\x^*)  \\
	& \geq -2D'\|\Bar{\z}^{k\mathcal{B}}-\z_i^{k\mathcal{B}} \|+f_i(\Bar{\z}^k\mathcal{B})-f_i(\x^*).
\end{aligned}
\end{equation}

Rearranging the terms above and summing from $t=0$ to $\infty$ completes the proof. 
	
\section{Proof of the convergence rate}


First we derive an intermediate proposition.
\begin{adjustwidth}{0.5cm}{0cm}\qquad
\begin{proposition}
For each $m$ and $ k$, the following inequalities hold:
\begin{enumerate}[(a)]
\item For $1 \leq i \leq n$, $\sum_{k=0}^{\infty} \alpha_k |z_{im}^{k\mathcal{B} }-\Bar{z}_m^{k\mathcal{B}}| < \infty$.
\item For $n+1 \leq i \leq 2n$, $\sum_{k=0}^{\infty} \alpha_k |z_{im}^{k\mathcal{B}}| < \infty$.
\end{enumerate}
\end{proposition}
\begin{proof}
Using the result of Lemma \ref{lemma3}(a), for $1 \leq i \leq n$, \begin{equation}
	\begin{aligned}
	\sum_{k=1}^T\alpha_k \|\z_i^{k\mathcal{B}}-\Bar{\z}^{k\mathcal{B}} \| & \leq \Gamma (\sum_{j=1}^{2n}\sum_{s=1}^d|z^0_{js}|)\sum_{k=1}^T \alpha_k \sigma^k 
	+\sqrt{d} n\Gamma D\sum_{k=1}^T\sum_{r=1}^{k-1}\sigma^{t-r}\alpha_k \alpha_{r-1}\\& +2\sqrt{d}D \sum_{k=0}^{T-1}\alpha_k^2.
	\end{aligned}
	\end{equation} 
	Applying inequality $ab \leq \frac{1}{2}(a+b)^2, a, b \in \mathcal{R}$,
	\begin{equation}
	\sum_{k=1}^T \alpha_k \sigma^k \leq \frac{1}{2}\sum_{k=1}^T (\alpha_k^2+ \sigma^{2k}) \leq \frac{1}{2}\sum_{k=1}^T \alpha_k^2 + \frac{1}{1-\sigma^2}
	\end{equation}
	\begin{equation}
	\begin{aligned}
	\sum_{k=1}^T\sum_{r=1}^{k-1}\sigma^{k-r}\alpha_k \alpha_{r-1} \leq \frac{1}{2}\sum_{k=1}^T\alpha_k^2 \sum_{r=1}^{r-1} \sigma^{k-r}  + \frac{1}{2}\sum_{r=1}^{T-1}\alpha_{r-1}^2 \sum_{k=r+1}^{T} \sigma^{k-r} \leq \frac{1}{1-\sigma}\sum_{k=1}^T \alpha_k.
	\end{aligned}
	\end{equation}
	Using the assumption that the step size satisfies $\sum_{k=0}^{\infty} \alpha_t^2 < \infty$ as $T \to \infty$, we complete the proof of part (a).
	The same techniques can be used to prove part (b).
\end{proof}
\end{adjustwidth}
We can now continue the proof of the stated convergence rate. Since the mixing matrices have columns that sum up to one we have $\bar{z}^{k\mathcal{B}+t'-1} = \bar{k\mathcal{B}},$
for all $t' = 1, \cdots, \mathcal{B}$. 

In the following step, we consider $t = k\mathcal{B}$ for some integer $k \geq 0$. Defining $f_{\min}:=\mathrm{min}_t f(\Bar{\mathbf{z}}^t)$, we have
\begin{equation}\label{eq:temp:proof}
\begin{aligned}
(f_{\min}-f^*)\sum_{t=0}^T\alpha_t &\leq \sum_{t=0}^T\alpha_t(f(\Bar{\mathbf{z}}^t)-f^* )\leq C_1+C_2\sum_{t=0}^T\alpha_t^2,
\end{aligned}
\end{equation}
where 
\begin{equation}
    C_1 =\frac{n}{2}(\|\Bar{\mathbf{z}}^0-\mathbf{x}^*\|^2-\|\Bar{\mathbf{z}}_{T+1}-\mathbf{x}^* \|^2) +D'\Gamma \sum_{j=1}^{2n}\frac{\|\mathbf{z}_j^0\|}{1-\sigma^2},
\end{equation}
\begin{equation}
C_2=\frac{nD'^2}{2}+4D'^2+D'\Gamma \sum_{j=1}^{2n}\|\mathbf{z}_j^0\|+\frac{2D'^2\Gamma}{1-\sigma}.
\end{equation}
Note that we can express \eqref{eq:temp:proof} equivalently as
\begin{equation}\label{eq:temp:proof2}
(f_{\min}-f^*)\leq \frac{C_1}{\sum_{t=0}^T\alpha_t}+\frac{C_2 \sum_{t=0}^T\alpha_t^2}{\sum_{t=0}^T\alpha_t}.
\end{equation}
Now, by recalling the statement of Assumption \ref{assumptionAll}(iii), we have that $\alpha_t=o(1/\sqrt{t})$. If we select the schedule of stepsizes according to $\alpha_t=1/\sqrt{t}$, the two terms on the right hand side of \eqref{eq:temp:proof2} satisfies 
\begin{equation}
   \frac{C_1}{\sum_{t=0}^T\alpha_t} = C_1 \frac{1/2}{\sqrt{T}-1}=\mathcal{O}(\frac{1}{\sqrt{T}}),
    \frac{C_2 \sum_{t=0}^T\alpha_t^2}{\sum_{t=0}^T\alpha_t}=C_2 \frac{\mathrm{ln}T}{2(\sqrt{T}-1)}=\mathcal{O}(\frac{\mathrm{ln}T}{\sqrt{T}}). 
\end{equation}
This completes the proof.